\documentclass{ws-ijfcs}

\usepackage{soulutf8} 
 \usepackage{cite} 
\usepackage{amssymb}
\usepackage{amsmath}
\usepackage{textcomp}
\usepackage{array}
\usepackage{color}
\usepackage{setspace}
\usepackage{comment,enumerate}
\usepackage[]{datetime}
 \usepackage{multirow}
\usepackage{array,float}
\usepackage{tikz,graphicx}
\def\block(#1,#2)#3{\multicolumn{#2}{|c|}{\multirow{#1}{*}{$ #3 $}}}

\usepackage{hyperref}
\newcommand{\dollar}[0]{\$}

\newcommand{\mymatrix}[2]{\left( \begin{array}{#1} #2\end{array} \right)}

\newcommand{\myvector}[1]{\mymatrix{c}{#1}}

\newcommand{\mypar}[1]{\left( #1 \right)}

\newcommand{\modm}{\mathtt{MOD_m}}

\newcommand{\Leq}{\mathtt{EQ}}
\newcommand{\modtwothree}{\mathtt{MOD23}}

\newcommand{\dyck}{\mathtt{DYCK}}

\newcommand{\abu}[1]{{\color{blue}#1}}
\newcommand{\ozlem}[1]{{\color{red}#1}}

\begin{document}

\title{New Results on Vector and Homing Vector Automata\footnote{Many of these results were obtained during Yakary\i lmaz's visit to Bo\u{g}azi\c{c}i University in July-August 2017.}}

\author{\"{O}zlem Salehi}

\address{ Bo\u{g}azi\c{c}i University, Department of Computer Engineering \\ 
Bebek 34342, \.{I}stanbul, Turkey\\
	\email{ozlem.salehi@boun.edu.tr}} 
    
\author{Abuzer Yakary{\i}lmaz}

\address{University of Latvia, Center for Quantum Computer Science\\ 
R\={\i}ga, Latvia \\
  \email{abuzer@lu.lv}
}
    
\author{A. C. Cem Say}

\address{ Bo\u{g}azi\c{c}i University, Department of Computer Engineering \\ 
Bebek 34342, \.{I}stanbul, Turkey\\
	\email{say@boun.edu.tr}} 
    
\markboth{\"{O}. Salehi, A. Yakary{\i}lmaz, A. C. C. Say}
{New Results on Vector and Homing Vector Automata}

\maketitle

\begin{abstract} 
We present several new results and connections between various extensions of finite automata through the study of vector automata and homing vector automata. We show that homing vector automata outperform extended finite automata when both are defined over $ 2 \times 2 $ integer matrices. We study the string separation problem for vector automata and demonstrate that generalized finite automata with rational entries can separate any pair of strings using only two states. Investigating stateless homing vector automata, we prove that a language is recognized by stateless blind deterministic real-time version of finite automata with multiplication iff it is commutative and its Parikh image is the set of nonnegative integer solutions to a system of linear homogeneous Diophantine equations.   
\end{abstract}

\section{Introduction}

Extending finite automata with an external storage is a well studied topic of automata theory. In this manner, various extensions of finite automata such as quantum and probabilistic models can be viewed as finite automata augmented with a vector as memory and the computation process can be seen as a series of vector and matrix multiplications. 

With the motivation of the matrix multiplication view of programming, we have previously introduced vector automata (VA) \cite{SYS13}, which are finite automata equipped with a vector that is multiplied with an appropriate matrix at each step. Only one of the entries of the vector can be checked for equivalence to a rational number and the computation is successful if it ends in an accept state with a successful equivalence test. We have shown that this model is closely related to various extensions of finite automata including counter automata \cite{FMR67}, finite automata with multiplication (FAM) \cite{ISK76}, and generalized finite automata \cite{Tu68}.

In many of the models which extend finite automata, a computation is successful if the register is equal to its initial value at the end of the computation. Applying this idea to vector automata, we have introduced the homing vector automaton (HVA) \cite{SS15}, which can check its vector for equivalence to its initial value and accepts a string if the computation ends with the initial vector in an accept state. We have shown that there exist close connections between homing vector automata \cite{SSD16} and extended finite automata \cite{DM00}. In fact, homing vector automata can be seen as a special case of rational monoid automata \cite{RK10}.

The close connections among the models exhibited so far motivates us to further investigate vector automata and homing vector automata. In this respect, the aim of this paper is to present some new results and establish further relationships among the various extensions of finite automata, through the study of these two models. We also provide some additional results about vector automata and homing vector automata, which may further advance the research in the topic.

We first focus on extended finite automata and show that allowing postprocessing does not change the computational power of the model. It is an open question whether the class of languages recognized by extended finite automata over $ 3 \times 3 $ integer matrices includes that of one-way nondeterministic blind three counter automata. We prove that one-way three-dimensional nondeterministic blind homing vector automata (NBHVA) with integer entries can simulate any nondeterministic blind multicounter machine. We answer an open question from \cite{SSD16}, by proving that one-way NBHVAs are more powerful than extended finite automata \cite{DM00} when both are defined on the set of $ 2 \times 2$ integer matrices.

Given two strings, a finite automaton (DFA) is said to separate them if it accepts one and rejects the other. The string separation problem, which asks for the minimum number of states needed for accomplishing this task, was introduced by Goral{\v{c}}{\'\i}k and Koubek \cite{GK86}, and the best known upper bound was given in \cite{Ro89}. The results from \cite{GK86} provide a logarithmic lower bound in the length of the strings. The problem of finding a generic tight bound is still open \cite{DESW11}. 

We study the problem of string separation, and show that deterministic blind vector automata with vectors of size two can distinguish any string from any other string in blind mode (without the need to check the entries of the vector during the computation) even when they are restricted to be stateless (i.e. with only one state). This result implies that generalized finite automata with rational entries can distinguish any pair of strings by using only two states. We also present some results on finite language recognition.

Stateless machines \cite{YDI08,IKO10,KMO09} have been investigated by many researchers, motivated by their relation to membrane computing and P systems  \cite{Pau00}, which are stateless models inspired from biology. While vector automata can simulate their classical states in their vectors by using longer vectors, this is not the case for homing vector automata. This leads us to investigate stateless homing vector automata in more detail. 

Our study on stateless homing vector automata yields a characterization for the class of languages recognized by stateless real-time deterministic FAMs without equality (0-DFAMW) \cite{ISK76}. It turns out that a language is recognized by a 0-DFAMW iff it is commutative and its Parikh image is the set of nonnegative solutions to a system of linear homogeneous Diophantine equations.  When the computation is nondeterministic, then any language recognized by a stateless real-time nondeterministic FAM without equality is commutative. We conclude by providing some further examples and observations about language recognition power of stateless homing vector automata.

\section{Background} \label{sec: back}

\subsection{Preliminaries}

The following notation will be used throughout the paper. The set of states is $Q = \{ q_1,\ldots,q_n \}$ for some $ n \geq 1 $, where  $ q_1 $ is the initial state, unless otherwise specified. $Q_a \subseteq Q$ is the set of accept state(s). The input alphabet is denoted by $\Sigma$ and we assume that it never contains $ \dollar $ (the right end-marker). For a given string $ w \in \Sigma^* $, $w^r$ denotes its reverse, $ |w| $ denotes its length, $ w[i] $ denotes its $i$'th symbol, $ |w|_{\sigma} $ denotes the number of occurrences of symbol $ \sigma $ in $ w $ and $ \mathtt{L}_w $ denotes the singleton language containing only $ w $. For a given language $L$, its complement is denoted by $\overline{L}$. The power set of a set $S$ is denoted $ \mathcal{P}(S) $.

For a string $ w \in \Sigma^* $ where $ \Sigma=\{\sigma_1,\sigma_2,\dots,\sigma_k\} $ is an ordered alphabet, the Parikh image of $ w $ is defined as $ \phi(w)=(|w|_{\sigma_1}, |w|_{\sigma_2},\dots, |w|_{\sigma_k}) $. For a language $ L $, $ \phi(L) = \{\phi(w)| w \in L\} $. 

A subset $ S \subseteq \mathbb{N}^n  $ is a \textit{linear} set if $ S=\{v_0 + \Sigma_{i=1}^k c_iv_i | c_1,\dots,c_k \in \mathbb{N}\} $ for some $ v_0,\dots,v_k \in\mathbb{N}^n  $. A \textit{semilinear} set is a finite union of linear sets. A language is called \textit{semilinear} if $ \phi(L) $ is semilinear.

For a machine model $A$, $\mathfrak{L}(A)$
denotes the class of all languages recognized by machines of type $A$. We denote the class of context-free languages by $ \mathsf{CF} $ and the set of recursively enumerable languages by $ \mathsf{RE} $.

 For a given row vector $ v $, $ v[i] $ denotes its $i$'th entry. Let $ A_{k \times l} $ be a $ k \times l $ dimensional matrix. $ A[i,j] $ denotes the entry in the $ i $'th row and $ j $'th column of $ A $. Given matrices $ A_{k\times l} $ and $ B_{m \times n}  $, their tensor product is defined as

\[
A \otimes B _{km \times ln} =\mymatrix{ccc}{ A[1,1]B & \cdots & A[1,l]B  \\ \vdots &  & \vdots \\ A[k,1]B &\dots &A[k,l]B}.
\]

\subsection{Machine definitions}

In this section, we give definitions of the various models which will appear throughout the paper. 

An input string $ w $ is placed on a tape in the form $ w\dollar$. Most of the time, we focus on real-time computation, where the tape head moves right at each step. A machine is \textit{deterministic} if the next move of the machine is uniquely determined, and \textit{nondeterministic} if there may exist more than one possible move at each step.

When we want to specify the number of states of a machine, we add an $ n $- (or $(n)$- to avoid any confusion) to the front of the model name where $ n $ is the number of states. We will examine stateless models, i.e. one-state automata where the single state is accepting, in more detail. It is clear that if the single state is non-accepting, then the automaton can recognize only the empty language. When a machine is stateless, we denote this by adding 0- to the front of the abbreviation for the model name.

A \textit{(real-time) deterministic $k$-counter automaton} (D$ k $CA) \cite{FMR67} is a deterministic finite automaton equipped with $ k $ counters. The counters are initially set to 0 and updated at each step by $c \in \{-1,0,1\}^k$, based on the current state, scanned symbol and the current status of the counters. The status of the counters is given by $\theta  \in \{=,\neq\}^k $, where $ = $ and $ \neq  $ denote whether the corresponding counter values equal zero or not, respectively. An input string is accepted if the machine is in an accept state at the end of the computation. A \textit{(real-time) deterministic blind $k$-counter automaton} \cite{Gr78} D$ k $BCA is a D$ k $CA which can check the value of its counters only at the end of the computation. An input string is accepted by a blind counter automaton if the machine is in an accept state and all counter values are equal to 0 at the end of the computation.

A \textit{generalized finite automaton} (GFA) \cite{Tu69} is  a
5-tuple $ G=(Q,\Sigma,\{A_{\sigma \in \Sigma}\}, v_0, f),$
where the $A_{\sigma \in \Sigma}$'s are $|Q|\times |Q|$ real valued
transition matrices, and $v_0$ and $f$ are the real valued initial row vector and
final column vector, respectively. The acceptance value for an input string $w
\in \Sigma^*$ is defined as $f_{G}(w)=v_oA_{w[1]}\cdots A_{w[|w|]}f$. 

There are various ways to define the language recognized by a generalized finite automaton. In this paper, we are interested in the class  $\textup{S}^=$, which contains languages of the form $L=(G,=\lambda)=\{w\in \Sigma^* \mid f_{\mathcal{G }}(w)=\lambda\}$, where $ \lambda \in \mathbb{R}$ is called the cutpoint \cite{Tu69}. 

A GFA whose components are restricted to be rational numbers is called a \textit{Turakainen finite automaton} (TuFA) in \cite{Yak12} and the class of languages recognized by a TuFA in the same language recognition mode with a rational cutpoint is denoted by $\textup{S}^=_{\mathbb{Q}}$.

A \textit{(real-time) deterministic finite automaton with multiplication} (DFAM) \cite{ISK76} is a deterministic finite automaton equipped with a register holding a positive rational number. The register is initialized to 1 at the beginning of the computation and multiplied with a positive rational number at each step, based on the current state, scanned symbol and whether the register is equal to 1 or not. An input string is accepted if the computation ends in an accept state with the register value being equal to 1, after processing the right end-marker. A DFAM \textit{without equality} (DFAMW) is a DFAM which cannot check whether or not the register has value 1 during computation. Nondeterministic versions NFAM and NFAMW are defined analogously.\footnote{The original definition of FAMs is given for \textit{one-way} machines (1DFAM, 1NFAM, 1DFAMW, 1NFAMW) where the tape head is allowed to stay on the same input symbol for more than one step. The computation halts and the string is accepted when the machine enters an accept state with the tape head on the end-marker $\dollar$.}

Let $ M $ be a monoid. An \textit{extended finite automaton over $ M $} (\textit{$ M $-automaton}, \textit{monoid automaton}) \cite{DM00, Co05,Ka09} is a nondeterministic finite automaton equipped with a register which holds an element of the monoid $ M $. The register is initialized with the identity element of $ M $ and multiplied (operation is applied) by an element of the monoid based on the current state and the scanned symbol (or the empty string). An input string is accepted if the value of the register is equal to the identity element of the monoid and the computation ends in an accept state. When $ M $ is a group, then the model is called \textit{group automaton}.

Note that extended finite automata are nondeterministic by definition and they are blind in the sense that the register cannot be checked until the end of the computation. Computation of an extended finite automaton is not real-time since the machine is allowed to make $ \varepsilon $ transitions. The class of languages recognized by $ M $-automata is denoted by $ \mathfrak{L}(M) $. 

\subsection{Vector automata and homing vector automata}

A \textit{ (real-time) $k$-dimensional deterministic vector automaton} (DVA($k$))
\cite{SYS13} is a 7-tuple

\[
V = (Q,\Sigma,M,\delta,q_1,Q_a,v_0),
\]
where $ v_0 $ is the initial ($k$-dimensional, rational-valued) row vector, $ M $ is a finite set of $k \times k$-dimensional rational-valued matrices, and $ \delta $ is the transition function (described below) defined as
\[
\delta:Q \times \Sigma \cup \{\dollar\} \times \{=,\neq\} \rightarrow Q \times M.
\]
Let $ w\in \Sigma^* $ be a given input. The automaton $ V $ reads the sequence $ w\dollar $ from left to right symbol by symbol. It uses its states and its vector to store and process the information. In each step, it can check whether the first entry of the vector is equal ($=$) to 1 or not ($ \neq $). We call this feature the ``status" of the first entry. 

The details of the transition function are as follows. When $ V $ is in state $ q \in Q $, reads symbol  $ \sigma \in \Sigma \cup \{\dollar\} $, and the first entry status is $ \tau \in \{ =,\neq \} $, the transition
\[
\delta(q,\sigma,\tau) = (q',A)
\]
results in $V$ entering state $ q' \in Q $, and its vector being multiplied by $ A \in M $ from the right.

At the beginning of the computation, $ V $ is in state $ q_1 $ and the vector is $ v_0 $. Then, after reading each symbol, the state and vector are updated according to the transition function as described above. The input $ w $ is accepted if the final state is an accept state and the first entry of the final vector is 1 after processing the right end-marker $\dollar$. Otherwise, the input is rejected. The set of all accepted strings is said to be the language recognized by $ V $.

A \textit{ (real-time)  $k$-dimensional deterministic homing vector automaton} (DHVA($k$)) is defined in \cite{SS15}
as being different from vector automata in two ways: (1) Homing vector automata do not read the right end-marker after reading the input, so there is no chance of postprocessing and, (2) instead of checking the status of the first entry, a homing vector automaton checks whether the complete current vector is identical to the initial vector or not. Formally, the transition function $\delta$ is defined as $$\delta: Q \times \Sigma  \times \{=,\neq\}  \rightarrow Q\times M,$$ where $ = $  indicates
equality to the initial vector $ {v_0} $, and $ \neq $ indicates inequality. An input string is accepted if the computation ends in an accept state and the vector is equal to its initial value. 

The blind versions of these models, \textit{(real-time) $k$-dimensional deterministic blind vector automaton} (DBVA($ k $)) and \textit{(real-time) $k$-dimensional deterministic blind homing vector automaton} (DBHVA($ k $)) cannot check the status of the vector until the end of the computation. Therefore, the domain of the transition function changes to $ Q \times \Sigma \cup \{\dollar\} $ and $ Q \times \Sigma$ for vector automata and homing vector automata respectively. The acceptance condition is the same as in the non-blind case.

The definitions of \textit{(real-time) $k$-dimensional nondeterministic vector automaton}, abbreviated NVA($k$), and \textit{(real-time) $k$-dimensional nondeterministic homing vector automaton}, abbreviated NHVA($k$), are almost the same as that of the deterministic versions, except that the range of the transition function $ \delta $ is now defined as $ \mathcal{P}(Q \times M)$, which allows the machine to follow more than one computational path. An input string is accepted if and only if there is a path ending with the acceptance condition.

Abbreviations used for some model variants discussed so far are given in Table \ref{tab: abb}.  

	\begin{table} \caption{Abbreviations for some model names.}\label{tab: abb}{\footnotesize \begin{tabular}{|p{2.1cm}|p{1.7cm}|p{1.7cm}|p{2.2cm}|p{2.2cm}|} 
		\hline  &  Deterministic  &  Deterministic & Nondeterministic & Nondeterministic \\ 
		&   blind & &  blind & \\ 
		\hline Vector \hspace{0.8cm} automaton & \hspace{1.7cm} DBVA($k$) &  \hspace{1.7cm} DVA($ $k$ $) & \hspace{1.7cm} NBVA($k$) &  \hspace{1.7cm}NVA($k$) \\ 
		\hline Homing vector automaton &\hspace{0.7cm} DBHVA($k$) & \hspace{1.7cm} DHVA($k$) &  \hspace{1.7cm}NBHVA($k$) &  \hspace{1.7cm}NHVA($k$) \\ 
		\hline Counter automaton &\hspace{0.7cm} D$k$BCA & \hspace{1.7cm} D$k$CA &  \hspace{1.7cm}N$k$BCA &  \hspace{1.7cm}N$k$CA \\ 
		\hline Finite automata \hspace{0.9cm}   with\hspace{1cm} multiplication & \hspace{6cm}DFAMW & \hspace{1.7cm} DFAM & \hspace{1.7cm} NFAMW & \hspace{1.7cm} NFAM \\ 
		\hline 
		\end{tabular}}
	\end{table} 

When we do not want to explicitly state the dimension of the vector or the number of the counters, then we may omit $ k $ in the abbreviation when we talk about vector automata, homing vector automata and counter automata. Similarly, we may omit  $ D  $ and $ N $ in the abbreviation when we talk about a statement that is true for both cases.

\section{New results on homing vector automata and extended finite automata} \label{sec: relation }
In this section, we will focus on the relationship between extended finite automata and homing vector automata.

\subsection{Some lemmas on homing vector automata}
\label{sec: endmarker}

We will start by proving a lemma that will be used in the rest of the section. 

Homing vector automata are not allowed to perform postprocessing by definition, since they do not read the right end-marker. In this section, we show that allowing postprocessing does not bring any additional power to nondeterministic blind homing vector automata, as the postprocessing step can be handled by using some extra states. We prove the result for the more general case of 1NBHVAs, NBHVAs that are capable of making $ \varepsilon $ transitions. The computation of a 1NBHVA using end-marker ends once the machine processes the end-marker $ \dollar $.

HVAs using end-marker will be denoted by the abbreviation $\textup{HVA}_\$ $.

\begin{lemma}
	\label{thm: NBHVA-endmarker}
	Let $ L $ be a language  recognized by an $(n)$-$\textup{1NBHVA}_\$(k)$ $ V $. Then, $ L $ is also recognized by an $(n+2)$-$\textup{1NBHVA}(k)$ $ V' $.
\end{lemma}
\begin{proof}
	We construct $V'$ such that $V'$ mimics the transitions of $V$ on every possible symbol $\sigma \in \Sigma \cup \{\varepsilon\}$. In addition, we create new transitions to handle the postprocessing, which emulate  $ V $'s last action before reading the end-marker (which would end up in an accept state) and the end-marker ($ \sigma \dollar $) at once: At any point during the scanning, if reading $ \sigma $ would cause $V$ to switch to a state from which the end-marker would lead to an accept state, a new nondeterministic transition takes  $ V' $ to the  additional state, which is an accept state. During this transition, the vector is multiplied by a matrix of the form $ A_\sigma A_\dollar  $, where $ A_\sigma $ and $ A_\dollar $ are the corresponding matrices defined in the transition function of $ V $. All other states of $V'$, which are inherited from $V$, are designated to be non-accept states. 

Thus, $ V' $ simulates the computation of $ V $ on any non-empty string, and accepts the input in some computational path if and only if $ V $ accepts it. 

If $ L $ contains the empty string, we add one more  state that has the following properties: (i) it becomes the new initial state of the resulting machine, (ii) it is an accept state, (iii) it causes the machine to behave (i.e. transition) like the original initial state of $ V $ upon reading the first symbol, and (iv) there is no transition coming in to this state.
\end{proof}

The idea given in the proof of Lemma \ref{thm: NBHVA-endmarker} does not apply for non-blind models since the status of the vector may be changed after reading the last symbol of the input (just before reading the right end-marker). In fact, one can show that DHVAs using end-marker are more powerful than ordinary DHVAs in terms of language recognition by the witness language $\mathtt{NEQ}=\{ a^ib^j : i \neq j\}$ .
\\
In the following lemma, we show that any 1NBHVA with end-marker whose matrices are rational valued can be simulated by integer valued matrices in the cost of increasing the size of the vector by 2. The lemma is also valid for the deterministic and real-time models.

\begin{lemma}
	\label{thm: ratint}
For any given rational-valued $(n)$-$ \textup{1NBHVA}_\$(k)$ $ V $, there exists an integer-valued $(n)$-$ \textup{1NBHVA}_\$(k+2)$ $ V' $  that recognizes the same language.
\end{lemma}
\begin{proof}
	It is easy to see that, for any HVA, the initial vector $ v_0 $ can be replaced with the vector $ t v_0 $ for any rational $ t \neq 0 $ without changing the recognized language. Therefore, any HVA can be assumed to be defined with an integer-valued initial vector without loss of generality.

	Let the initial vector of $ V $ be
	$
	v_0 = \mypar{ i_1 ~~ i_2 ~~ \cdots ~~ i_k } \in \mathbb{Z}^k$. Furthermore, let	$ S_{i_a}, S_{i_b},S_{i_\varepsilon}$ and $S_{i_\dollar} $
	be the sets of $ k \times k $ rational-valued matrices of $ V $ to be applied  when $ V $ is in state $q_i$ and reads the symbols $ a $, $ b $, $ \varepsilon $ and $\dollar$, respectively. 
	
	The automaton $ V' $ is obtained from $V$ by modifying the initial vector and matrices.  We pick an integer $ c \in \mathbb{Z}$ such that when multiplied with $ c $, the matrices in $S_{i_a},S_{i_b},S_{i_\varepsilon}$ and $ S_{i_\dollar} $ contain only integer values, where $ 1 \leq i \leq n $. Then, we define the set $  S'_{i_\sigma} $ by letting
	\[
	v_0' = \mypar{ v_0 ~~ 1 ~~ 1 } \mbox{ ~and~ }
	A'_{i_{\sigma_j}} = \mymatrix{c|cc}{ c A_{i_{\sigma_j}} & 0 & 0 \\ \hline 0 & c & 0 \\ 0 & 0 & 1 },
	\]
	where $  A_{i_{\sigma_j}} \in  S_{i_{\sigma}} $, $ \sigma \in \{a,b,\varepsilon,\dollar\} $ and $ j $ enumerates the number of possible transitions for $ \sigma $ in state $ q_i $. $ v'_0 $ is the initial vector of $ V' $, and when in state $ q_i $, $V'$ multiplies the vector with an integer-valued matrix from the set $ S_{i_{a}}'$, $ S_{i_{b}}' $ or $ S_{i_{\varepsilon}}' $ upon reading the inputs $a$, $b$, and $ \varepsilon $ respectively. When in state $q_i$ and reading symbol $\dollar$, the vector is multiplied with a matrix from the set $  S''_{i_\dollar} $, which is obtained by multiplying the matrices in $ S_{i_\dollar}' $ with
	\[
	\mymatrix{ c|cc }{
		& 0 & 0 \\
		~~~ - I ~~~~~ &  \vdots & \vdots
		\\
		& 0 & 0
		\\ \hline
		i_1 ~~ i_2 ~~ \cdots ~~ i_k & 0 &  0 \\ 
		i_1 ~~ i_2 ~~ \cdots ~~ i_k & 1 &  1
	},
	\]
	
	where $ I $ denotes the identity matrix. 
	
	Let $ w \in \{a,b\}^* $ be a string of length $ l $, and let's consider a single computation path of length $ p+1 \geq l+1 $ for $ w $. The final vector is calculated as
	\[
	v_f = v_0 M_1  M_2 \cdots   M_{p} M_{p+1} 
	=
	\mypar{ j_1 ~~ j_2 ~~ \cdots ~~ j_k } \in \mathbb{Q}^k ,
	\]
	where $ M_i $ is the matrix used at the $ i $'th transition step.
	It is easy to see that 
	\[
	v_f' =  v_0' M_1' \cdots M_p'   M_{p+1}' 
	=
	\mypar{ c^{p+1} v_f ~~~ c^{p+1} ~~~ 1 }.
	\]
	(Note that the  equation above contains $ M_{p+1}' $, but not $ M_{p+1}'' $.) By setting $  c'= c^{p+1} $, we can rewrite $ v_f' $ as
	\[
	\mypar{ c'j_1 ~~~  c'j_2~~~ \cdots ~~~  c'j_k~~~ c'~~~ 1 }
	\in \mathbb{Z}^{k+2}.
	\]
	
	For accepted input strings,  $ v_f' $ holds $ c'$ times the initial vector. The postprocessing step, which is accomplished by a matrix from the set $  S''_{i_\dollar} $, helps setting the vector back to its initial value to satisfy the acceptance criteria.
	
	Hence, $ v_f'' = v_0' M_1' \cdots M_p'   M_{p+1}'' $ is
	\[
	v''_f
	=
	\mypar{  c'i_1-  c'j_1+ i_1  ~~~~ 
		c'i_2 -  c'j_2 + i_2 ~~~~
		\cdots
		~~~~  c'i_k-  c'j_k + i_k
		~~~~ 1 ~~~~ 1}.
	\]
	
	It is clear that $ v_0 = v_f $ if and only if $ v'_0 = v''_f $. Thus, rational-valued $ V $ and integer-valued $ V' $ recognize the same language.
\end{proof}

\begin{corollary}
	Rational-valued \textup{1NBHVA}s and integer-valued \textup{1NBHVA}s recognize  the same class of languages.
\end{corollary}
\begin{proof}
	By using Lemma \ref{thm: ratint}, we can conclude that any rational-valued 1NBHVA can be simulated by an integer-valued 1NBHVA using the end-marker. Then, by using Lemma \ref{thm: NBHVA-endmarker}, we can remove the end-marker.
\end{proof}

\subsection{Extended finite automata and homing vector automata defined over integer matrices}

Before presenting the main results of the section, we observe that postprocessing does not bring additional power to extended finite automata.

\begin{lemma}
Let $ L $ be a language recognized by an extended finite automaton $ V $ using end-marker. Then $ L $ is also recognized by an extended finite automaton $ V' $ that does not use the end-marker.
\end{lemma}
\begin{proof} 
	Extended finite automata are one-way and nondeterministic by definition. Hence, the proof idea of Lemma \ref{thm: NBHVA-endmarker} applies here as well. With the help of the additional transitions and extra states, the postprocessing step can be handled by an ordinary extended finite automaton without end-marker.  
\end{proof}

Let $ M $ be a multiplicative monoid of $ k \times k $ integer matrices. We denote by $\textup{HVA}(k)_{M} $ a $ k $-dimensional homing vector automaton whose matrices belong to $ M $. For the rest of the section, we are going to investigate $ M $-automata and $\textup{HVA}(k)_{M} $. $ \mathbb{Z}^{k \times k} $ denotes the multiplicative monoid of $ k \times k$ matrices with integer entries. 

In \cite{SDS16}, it is shown that $\mathsf{RE} \subseteq \mathfrak{L}\textup{(1NBHVA(4)}_{\mathbb{Z}^{4 \times 4}})$. Therefore, Lemma \ref{thm: ratint} says something new only about the class of languages recognized by 3-dimensional 1NBHVAs with integer entries. Note that any blind $ k $-counter automaton can be simulated by a one-dimensional homing vector automaton whose register is multiplied by $ k $ distinct primes and their multiplicative inverses. Using the fact that any 1N$ k $BCA can be simulated by a 1NBHVA(1), we show that the same result can be achieved by a 1NBHVA(3) whose matrices are restricted to have integer entries. Note that the result is also true for real-time NBHVAs.

\begin{theorem}\label{cor: z3}
	$ \bigcup_k\mathfrak{L}(\textup{1N$ k $BCA}) \subsetneq \mathfrak{L}\textup{(1NBHVA(3)}_{\mathbb{Z}^{3 \times 3}}) $.
\end{theorem}
\begin{proof}
	Any 1N$ k $BCA can be simulated by a 1NBHVA(1)$ _\mathbb{Q^+} $ and any 1NBHVA(1)$ _\mathbb{Q^+} $ can be simulated by a 1NBHVA$_ \dollar $(3)$ _{\mathbb{Z}^{3 \times 3}} $ by Lemma $ \ref{thm: ratint} $. By using additional states, we can obtain an equivalent $\textup{1NBHVA(3)}_{\mathbb{Z}^{3 \times 3}} $ without end-marker by Lemma \ref{thm: NBHVA-endmarker}. The inclusion is proper since the unary nonregular language $  \mathtt{UPOW'=\{a^{n+2^n}\} | n\geq 0} $  can be recognized by a $\textup{NBHVA(3)}_{\mathbb{Z}^{3 \times 3}} $ \cite{SDS16}, which is not possible for 1N$ k $BCAs \cite{Gr78,Ib78}.
\end{proof}	

It is an open question whether $ \mathfrak{L}(\mathbb{Z}^{3 \times 3}) $ includes the class of languages recognized by 1N3BCAs. We cannot adapt Theorem \ref{cor: z3} to $ \mathbb{Z}^{3 \times 3} $-automata since the product of the matrices multiplied by the register in Lemma \ref{thm: NBHVA-endmarker} is not equal to the identity matrix and the acceptance condition for extended finite automata is not satisfied.

In \cite{SSD16} it is shown that the non-context-free language $ \mathtt{POW_r}=\{a^{2^n}b^n | n \geq 0\} $ can be recognized by a $ \textup{DBHVA(2)}_{\mathbb{Z}^{2 \times 2}} $. It is left open whether there exists a 1NBHVA with integer entries recognizing $ \mathtt{POW_r} $. Theorem \ref{thm: z22} answers two open questions from \cite{SSD16}, by demonstrating a 1NBHVA with integer entries recognizing $ \mathtt{POW_r} $ and revealing that 1NBHVAs are more powerful than the corresponding extended finite automata, when the matrices are restricted to the set $ \mathbb{Z}^{2 \times 2} $. By $ \mathbf{F}_2 $ we denote the free group of rank 2.

\begin{theorem}\label{thm: z22}
	$ \mathfrak{L}(\mathbb{Z}^{2 \times 2}) \subsetneq \mathfrak{L}(\textup{1NBHVA(2)} _{\mathbb{Z}^{2 \times 2}}) $.
\end{theorem}
\begin{proof} In \cite{SSa18}, it is proven that any $ \mathbb{Z}^{2 \times 2} $-automaton can be converted to a $ \mathbf{F}_2 $-automaton. By setting the initial vector to be $ (1~~0) $, any $ \mathbf{F}_2 $-automaton can be simulated by a 1NBHVA(2) whose matrices have determinant 1 with integer entries \cite{SSD16} and the inclusion follows.

	Now we are going to prove that the inclusion is proper. Since $ \mathbf{F}_2 $-automata recognize exactly the class of context-free languages \cite{DM00,Co05}, $ \mathfrak{L}(\mathbb{Z}^{2 \times 2})= \mathsf{CF}$. Let us construct a DBHVA$ _\dollar $(2)$ _{\mathbb{Z}^{2 \times 2}} $ $ V $ recognizing the non-context-free language $ \mathtt{POW_r}=\{a^{2^n}b^n | n \geq 0\} $. The state diagram of $ V $ is given in Figure \ref{fig: machine}.

	\begin{figure}\caption{State diagram of $ V $ accepting the language $\mathtt{POW_r}=\{a^{2^n}b^n | n \geq 0\} $. }\label{fig: machine}
		\centering
		
		\includegraphics[width=0.5\linewidth]{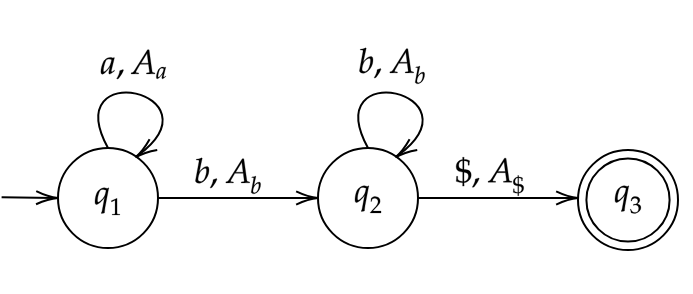}
	
	\end{figure}
	\[ 
	A_a=\mymatrix{ cc }{
		1 & 0 \\
		1 & 1
	}  ~~
	A_b=\mymatrix{ cc }{
		1 & 0 \\
		0 & 2
	}~~
	A_{\dollar}=\mymatrix{ rr }{
		1 & 1 \\
		-1 & -1
	} 
	\]
	The initial vector of $ V $ is $ v_0=(1~~1) $. While reading $ a $ in $ q_1 $, $ V  $ multiplies its vector with the matrix $ A_a $. It moves to $ q_2 $ when it scans the first $ b $ and multiplies its vector with the matrix $ A_b $ as it reads each $ b $. When $ V $ reads the end-marker, it multiplies its vector with the matrix $ A_{\dollar} $ and moves to $ q_3 $.

	When the vector is multiplied by $ A_a $, the first entry of the vector is increased by 1 and when the vector is multiplied by $ A_b $, the second entry of the vector is multiplied by 2. Hence, after reading an input string of the form $ a^ib^j $, the vector is equal to $ (i+1~~2^j) $, as a result of the multiplication by the matrix product $ {A_a}^i{A_b}^j  $. After multiplication by $ A_{\dollar} $, the second entry of the vector is subtracted from the first entry and this value is stored in both entries of the vector. The value of the final vector is equal to $ (1~~1) $ iff $  i+1-2^j=1 $. Hence, the accepted strings are those which are of the form $ a^{2^j}b^{j} $. As it is proven in Lemma \ref{thm: NBHVA-endmarker} that a NBHVA without end-marker recognizing the same language as a given NBHVA with end-marker can be constructed just by increasing the number of states but not the vector size, we can conclude that $ \mathtt{POW_r} $ can be recognized by a $ \textup{NBHVA(2)} _{\mathbb{Z}^{2 \times 2}} $. 

\end{proof}

\section{The string separation problem}
\label{sec:distinguish}

In this section we investigate the string separation problem for vector and homing vector automata. Recently, the same question has been investigated in \cite{BMY16,BMY17} for different models such as probabilistic, quantum, and affine finite automata (respectively, PFA, QFA, AfA). (We refer the reader to \cite{SayY14,DCY16A} for details of these models.) As these models are capable of storing information in their probabilities or amplitudes, they can be much more state-efficient than DFAs. For example, nondeterministic QFAs and zero-error AfAs can distinguish any pair by using only two states. 

The storage mechanism of vector automata and homing vector automata enables these two models to separate strings by encoding the string in their vector entries. When measuring efficiency of these two models, it is fairer to measure succinctness in terms of both the number of states and the vector size.\footnote{In \cite{SYS13}, the size of an $ n $-DBVA($ k $) is defined as $ nk $, and a hierarchy theorem is given based on the size of DBVAs.} 

We first show that VAs can simulate their classical states by using longer vectors and conclude that their stateless counterparts have the same computational power. 

\begin{lemma}
	\label{lem:1-state-VA}
	Any given $n$$\textup{-DVA($k$)}$ $ V $ can be simulated by a $\textup{0-DVA($nk$+1)}$, which recognizes the same language. 
\end{lemma}
\begin{proof}
	Note that the following modification to any DVA does not change the recognized language: If the last transition yields a non-accept state, then the vector is multiplied by the zero matrix, i.e. we guarantee that the first entry of the vector is set to 0 if the final state is non-accepting. Assume that $V$ is an $ n $-state $k$-dimensional DVA with this property.
	
	We describe a 0-DVA$(nk+1)$ $ V' $. We will call the vector of $V'$ `the new vector', and index its last $ nk $ entries by pairs as 
	\[
	\underbrace{(q_1,1), (q_1,2), \ldots, (q_1,k)}_{1st~block}, \underbrace{(q_2,1), (q_2,2), \ldots, (q_2,k)}_{2nd~block}, \ldots,\ldots, \underbrace{(q_n,1), (q_n,2), \ldots, (q_n,k)}_{n\mbox{-}th~block},
	\]
	where we have $ n $ blocks, and the $i$'th block is associated with $q_i \in Q$ ($1 \leq i \leq n$). Based on this representation, $V'$ will keep track of the state and the vector of $ V $ at each step. The remaining first entry of the new vector holds the sum of all of its entries at indices of the form $ (q,1) $, for any $q$. 
	
	During the simulation, all but one of the blocks will be forced to contain zeros in all of its entries. The only exception is the block corresponding to the current state of  $ V $, which will contain the values in $V$'s vector.
	Thus, the first entry of the new vector keeps the value of the first entry of $ V $'s vector for each step, and so, $V'$ can easily implement the vector test done by $ V $, i.e. testing whether the first entry of $V$'s vector is 1 or not.  
	
	The new vector is initially set to
	\[
	( v_0[1], v_0[1], v_0[2],\ldots,v_0[n],0,\ldots,0 ),
	\]
	where the first entry is set to $ v_0[1] $ and then the first block is set to $ v_0 $ since $ V $ starts its computation in $ q_1 $, and, all other blocks are set to zeros.
	
	Since $ V' $ is stateless, it has a single  $ (1+nk) \times (1+nk)$ matrix for every ($ \sigma $, $\tau$) pair.
	This big matrix is designed so that 
	each transition of the form $ \delta(q_i,\sigma,\tau) = (q_j,A) $ of $ V $
	is implemented by multiplying the $ i $'th block by $ A $ and transferring the result  to the $ j $'th block. Thus, both the simulated state and the vector of $ V $ are updated accordingly. Moreover, the first entry of the new vector is set to the summation mentioned above. In this way, the computation of $ V $ is simulated by $ V' $ and both of them make the same decision on each input, i.e. the first entry of the new vector is set to 1 if and only if the final state of $ V $ is accepting and the first entry of $ V $'s vector is 1.
\end{proof}

Stateless models are advantageous, as they allow leaving behind the details of the state transitions and focusing on the vector matrix multiplications. In the following series of results, we are going to show that stateless DBVAs can separate any pair of strings.

In \cite{SYS13}, it is proven that  $ \bigcup_k\mathfrak{L}$(DBVA($ k $)) = $S^=_{\mathbb{Q}}$. The proof involves construction of a 0-DBVA($ k $) to simulate a TuFA with $ k $ states. Therefore, any result about 0-DBVA($ k $)s is also true for TuFAs with $ k $ states and more generally for generalized finite automata with $ k $ states.

We start by distinguishing unary strings.

\begin{theorem}
	For any  $ i \geq 0 $, there exists a \textup{0-DBVA(1)} that distinguishes the string $ a^i $ from any other string in $\{a^*\}$.
\end{theorem}
\begin{proof}
	The initial vector is $ v_0=(2^i) $. For each symbol $a$, the vector is multiplied with $\bigl ( \frac{1}{2} \bigr) $. For any other input symbol, the vector is multiplied by (0). The end-marker is associated with multiplication by (1). Therefore, the final vector is $ (1) $ if and only if the automaton reads $ a^i $.
\end{proof}

\begin{corollary}
	The language $ \mathtt{L}_\varepsilon $ is recognized by a \textup{0-DBHVA(1)}
	.
\end{corollary}
\begin{proof}
	If one removes the transition associated with the end-marker from the machine described in the proof above for $i=0$, one obtains a \textup{0-DBHVA(1)} recognizing $ \mathtt L_\varepsilon $.
\end{proof}

For any $ x \in \{1,2\}^+ $, let $ e(x) $ be the number encoded by $ x $ in base 3:
\[
e(x) = 3^{|x|-1} x[1] + 3^{|x|-2} x[2] + \cdots + 3^1 x[|x|-1] + 3^0 x[|x|]  .
\]

The encoding $e(x)$ can be easily obtained by using vector-matrix multiplications. Starting with the initial vector $v_0= (1~~~0 )$, and multiplying the vector with 
\[
A_1 = \mymatrix{rr}{ 1 & ~ 1\\ 0 & ~ 3 }
\mbox{ and }
A_2 = \mymatrix{rr}{ 1 & ~ 2 \\ 0 & ~ 3},
\]
respectively
for each scanned $1$ and $2$ from left to right, one obtains $(1~~~e(x) )$ at the end.
We can easily extend this encoding for any generic alphabet. For any $ x \in \{1,2,\ldots,m-1\}^+ $ for some $m>3$, $e_m(x)$ is the number encoded by $ x $ in base $m$:
\[
e_m(x) = m^{|x|-1} x[1] + m^{|x|-2} x[2] + \cdots + m^1 x[|x|-1] + m^0 x[|x|]  .
\]
To encode $e_m(x)$, $v_0$ is multiplied with 
\[
A_i = \mymatrix{rr}{ 1 & ~ i \\ 0 & ~ m },
\]
for each symbol $i$. Note that the results proven in this section for binary alphabets can be easily extended to any alphabet using the encoding $ e_m(\cdot ) $.

Now we show that 2-dimensional vector automata can distinguish strings without using any states.

\begin{theorem}
	\label{thm: VA-dist-x}
	The string $ x \in \{1,2\}^+ $ is distinguished from any other string in $\{1,2\}^*$ by a \textup{0-DBVA(2)}.
\end{theorem}
\begin{proof}
	We build the machine in question. The initial vector  is
	\[
	v_0 = (1 ~~~ e(x^r)).
	\]
	
	The matrices applied to the vectors for symbols $ 1 $ and $ 2 $ are 
	\[
	A_1^{-1} = \mymatrix{rr}{ 1 & ~ -\frac{1}{3} \\ 0 & ~ \frac{1}{3} }
	\mbox{ and }
	A_2^{-1} = \mymatrix{rr}{ 1 & ~ -\frac{2}{3} \\ 0 & ~ \frac{1}{3} },
	\]
	respectively. When applied on the vector, the matrix $ A_1^{-1} $ (resp., $ A_2^{-1} $) leaves the first entry unchanged at value 1, and subtracts 1 (resp., 2) from the second entry and then divides the result of this subtraction by 3.

	Let $ y \in \{1,2\}^* $ be the given input. Suppose that $y \neq \varepsilon$. After reading $ y $, the vector changes as
	\[
	( 1 ~~~ e(x^r))
	\rightarrow
	\mypar{1 ~~~ \frac{e(x^r) - y[1]}{3}}
	\rightarrow
	\mypar{ 1 ~~~ \frac{e(x^r) - y[1]-3y[2]}{3^2} }
	\]
	\[
	\rightarrow
	\mypar{ 1 ~~~ \frac{e(x^r) - 3^0y[1]-3^1y[2]-3^2y[3]}{3^3} }     
	\]
	for the first three steps, and, at the $ i $'th step, we have
	\[
	\mypar{ 1 ~~~ \frac{e(x^r) - \sum_{j=1}^i 3^{j-1}y[j]}{3^i} }
	\rightarrow
	\mypar{ 1 ~~~ \frac{ \frac{e(x^r) - \sum_{j=1}^i 3^{j-1}y[j]} {3^i} - y[i+1]}{3} },
	\]
	which is equal to
	\[
	\mypar{ 1 ~~~ \frac{e(x^r) - \sum_{j=1}^{i+1} 3^{j-1}y[j]}{3^{i+1}} }.
	\]
	Thus, after reading $ y $, the vector is equal to
	\[
	\mypar{1 ~~~ e(x^r) - e(y^r)}.
	\]
	After reading the end-marker, the final vector is set as
	\[
	v_f = \mypar{1 + e(x^r) - e(y^r) ~~~ e(x^r)}
	\]
	by using the matrix
	\[
	A_{\$}=\mymatrix{cc}{ 1 & ~e(x^r) \\ 1 & 0 }.
	\]
	
	Thus, if $ x=y $, then $ v_f = v_0 $ and the first entry is 1. Otherwise, the first entry of $ v_f $ is not equal to 1.
	
	If $y=\varepsilon$, then the final vector is obtained by multiplying the initial vector with $A_{\$}$ upon reading the end-marker: $$v_f=\mypar{1 + e(x^r) ~~~ e(x^r)}.$$
	Since the first entry of the vector is never equal to 1, $\varepsilon$ is never accepted.
	
\end{proof}

We now focus on recognizing finite languages containing multiple strings. Stateless VAs can accomplish this task while no stateless HVA can recognize any finite language as we show in Section \ref{sec: stateless}.

\begin{theorem}\label{thm: VA-fin}
	Any finite language $ X =\{ x_1,\ldots,x_k \}  \subseteq \{1,2\}^+ $
	can be recognized by a \textup{0-DBVA($2^{k}+1$)}.
\end{theorem}
\begin{proof}
	We know from  the proof of Theorem \ref{thm: VA-dist-x} 
	that for any $ x_i \in X $, there is a 0-DBVA(2) $ V_{x_i} $ that starts in $ (1~~~e(x_i^r)) $ and ends in $ ( 1+e(x_i^r)-e(y^r) ~~~ e(x_i) ) $ for the input string $ y $. Let $ V'_{x_i} $ be a  0-DBVA(2) obtained by modifying  $ V_{x_i} $ such that the   vector equals $ (e(x_i^r)-e(y^r) ~~~ 0) $ at the end. 
	We build a \textup{0-DBVA($2^{k}+1$)} $ V_X $ that
	executes each $ V'_{x_i} $ in parallel (by employing a tensor product of  all those machines) in the last $ 2^k $ entries of its vector, and then performs an additional postprocessing task. The initial vector of $ V_X $ is 
	\[
	v_0=(1~~~\otimes_{i=1}^k (1~~~e(x^r_i)) ),
	\]
	that is, each $ V'_{x_i} $ is initialized with its initial vector. In order to execute each $ V'_{x_i} $,  the  matrices 
	\[
	\mymatrix{c|c}{1 & 0 \cdots 0 \\ \hline 0 & \\ \vdots & \otimes_{i=1}^k A_1^{-1} \\ 0 &  }
	\mbox{ and }
	\mymatrix{c|c}{1 & 0 \cdots 0 \\ \hline 0 & \\ \vdots & \otimes_{i=1}^k A_2^{-1} \\ 0 &  }
	\]
	are respectively used for symbols  $ 1 $ and $ 2 $. Before reading $ \dollar $, the vector is equal to
	\[ 
	(1~~~\otimes_{i=1}^k (1~~~e(x^r_i)- e(y^r) ).
	\] The transition matrix used upon scanning the end-marker symbol $\dollar$ can be described as the product of two matrices $ A_{\dollar_1} $ and $ A_{\dollar_2} $:  
	
	\[ A_{\dollar_1} =
	\mymatrix{c|c}{1 & 0 \cdots 0 \\ \hline 0 & \\ \vdots & \otimes_{i=1}^k A_0 \\ 0   }
	\mbox{ and }
	A_{\dollar_2} =
	\mymatrix{c|c}{1 & 1~~~\otimes_{i=1}^k (1~~~e(x^r_i)) \\ \hline 1 & \\0 & \\ \vdots &   0\\ 0 &  }
	\]
	where 
	\[ A_{0} =
	\mymatrix{cc}{0 & 0 \\ 1 &  0  }.
	\]
	
	After multiplication with $ A_{\dollar_1} $, the vector becomes
	\[
	(1 ~~~\otimes_{i=1}^k(e(x_i^r)-e(y^r)~~~0 )),
	\]
	which is identical to
	\[
	\mypar{1 ~~~ \Pi_{i=1}^k(e(x_i^r)-e(y^r)) ~~~ 0 ~~~ \cdots ~~~ 0 }.
	\]
	If $ y \in X $, then the value of the second entry in this vector is also equal to 0 since one of the factors is zero in the multiplication. Otherwise, the value of the second entry is non-zero. $ A_{\dollar_2} $ in the product for $\dollar$  sets the final vector to
	\[
	v_f = \mypar{1+\Pi_{i=1}^k(e(x_i^r)-e(y^r))~~~\otimes_{i=1}^k (1~~~e(x^r_i))}.
	\]

	Here,  the last $ 2^k-1 $ entries can be  set to those values easily, since $ v_0 $ is a fixed a vector. Thus, if $ y \in X $, then $ v_f = v_0 $ with $ v_f[1] = 1 $, and so the input is accepted. Otherwise, $ v_f \neq v_0 $ with $ v_f[1] \neq 1 $, and so the input is rejected.  
	
\end{proof}

Now we move on to string separation by homing vector automata. We show that any nonempty string can be distinguished from any other string by a 2-DBHVA($2$).

\begin{theorem}
	\label{thm:HVA-dist-sb}
	The string $ x \in \{1,2\}^+ $ is distinguished from any other string in $\{1,2\}^*$ by a \textup{2-DBHVA($2$)}.
\end{theorem}
\begin{proof}
	Let us construct a 2-DBHVA(2) distinguishing $x$. The initial state is named $ q_1 $, and the only accept state is named $ q_2 $.  We are going to use the encoding described above and the associated matrices. The proof is similar to the proof of Theorem \ref{thm: VA-dist-x}. The initial vector is $$ 	v_0=\mypar{1 ~~ 0}.$$   
	Let $ y \in \{1,2\}^* $ be the given input. If $ y = \varepsilon $, then the state is never changed, and so the input is never accepted. Therefore, we assume $ y \in \{1,2\}^+ $ in the remaining part. 
	
	To encode $x^r$, we will use the matrix $$ A_{x^r} = A_{x[|x|]}A_{x[|x|-1]}   \cdots A_{x[2]} A_{x[1]}. $$ 
	When in $ q_1 $, the vector is updated as $$ v_1=v_0 A_{x^r} A_{y[1]}^{-1} $$ upon reading the first symbol of $y$, and the state is set to $ q_2 $. The value of the vector is equal to $$\mypar{1~~~~ \frac{e(x^r)-y[1]}{3}}.$$ From now on, the state is never changed, but the vector is updated by multiplication with $ A_1^{-1} $ and $A_2^{-1}$ for each scanned symbol 1 and 2 respectively. Thus, as in Theorem \ref{thm: VA-dist-x}, after reading $y$ the vector is equal to 
	$$ v_f=(1 ~~e(x^r)-e(y^r)).$$
	
	Hence, we conclude that $v_f=v_0$ iff $x=y$.    
\end{proof}

It turns out that using two states in the above proof is essential, and we cannot obtain a similar result by using just one state, regardless of the dimension of the vector. More precisely, if $ x $ is accepted by a 0-NHVA, then $ xx $ is also accepted by the same automaton, since any further repetition of the same input naturally carries the vector through a trajectory that ends up in its initial value. We are therefore motivated to examine the limitations imposed by statelessness on homing vector automata in more detail in the next section.

The stateless vector automata we propose in Theorem \ref{thm: VA-fin} for recognizing a finite language require a vector of exponential size in the number of strings. Note that the machine $ V_X $ built in the proof of the theorem is deterministic and homing with end-marker. If we allow nondeterminism, then Lemma \ref{thm: NBHVA-endmarker} helps us obtain the same result with just two dimensions and using three states for ordinary NBHVAs.

\begin{corollary}
	Any	finite language $ X=\{x_1,\dots, x_k\} \subseteq \{1,2\}^+ $ can be recognized by a \textup{3-NBHVA(2)} $H_X$.
\end{corollary}
\begin{proof}
	We first construct a 0-$\textup{NBHVA}_\dollar(2)$ $N_X$. $N_X$ nondeterministically picks an $ x_i $ and then executes the deterministic automaton given in the proof of Theorem \ref{thm: VA-dist-x}. Since $N_X$ is homing using end-marker, we conclude that there exists a 
	3-NBHVA(2) $H_X$ recognizing $X$ by Lemma \ref{thm: NBHVA-endmarker}. (A 2-NBHVA(2) is in fact sufficient if $X$ does not contain the empty string, as can be seen in the proof of Lemma \ref{thm: NBHVA-endmarker}.)
\end{proof}

\section{Stateless homing vector automata}
\label{sec: stateless}

In this section, we investigate the computation power of stateless HVAs. In these machines, only the vector is updated after reading an input symbol.

\subsection{Remarks}
The limitation of having a single state for homing vector automata leads to the acceptance of the string  $ xx $, whenever the string $ x $ is accepted. This is true since further repetition of the same input naturally carries the vector through a trajectory that ends up in its initial value. Based on this observation, we can list the following consequences:

\begin{itemize}
	\item 		If string $ x $ is accepted by a \textup{0-NHVA} $H$, then any member of $ \{x\}^* $ is also accepted by $H$.
	
	\item 	If all members of a language $L $ are accepted by a \textup{0-NHVA} $H$, then any string in $ L^* $ is also accepted by $H$.
	
	\item 	If language $L $ is recognized by a \textup{0-NHVA}, then $ L = L^* $.
	
	\item 	\textup{0-NHVA}s cannot recognize any finite language except $ \mathtt{L}_\varepsilon $.
\end{itemize}

All of the  results given above are also valid for deterministic or blind models. We can further make the following observation.

\begin{lemma}
	\label{lemma: diff}
	If the strings $ w_1 $ and $ w_1w_2 $ are accepted by a \textup{0-DHVA} $ H$, then the string $ w_2$ is also accepted by $ H $.
\end{lemma}
\begin{proof}
	After reading $ w_1 $, the value of the vector is equal to its initial value. Since $ w_1w_2 $ is also accepted by $ H $, reading $ w_2 $ results in acceptance when started with the initial vector.
	
\end{proof}

For the unary case we have the following.
%

\begin{theorem}\label{thm: gcd}
	If the strings $ a^i $ and $ a^j $ ($ 1<i<j $) are accepted by a \textup{0-DHVA} $ H$, then the string $ a^{\gcd(i,j)} $ is also accepted by $ H $.
\end{theorem}
\begin{proof}
	It is well known that for any positive integers $ i,j $, there are two integers $ l_i $ and $ l_j $ such that $ i l_i + j l_j = \gcd(i,j) $. Assume that $ l_i $ is positive and $ l_j $ is non-positive. (The other case is symmetric.) Note that $ il_i \geq j(-l_j) $. The strings $ a^{j(-l_j)} $ and $ a^{i l_i} $ are accepted by $ H $. By Lemma \ref{lemma: diff}, the string  $ a^{i l_i-j(-l_j)} $,  which is $ a^{\gcd(i,j)} $, is also accepted by $ H $. 
	
\end{proof}

\begin{corollary}
	\label{cor: gcd}
	If the strings $ a^i $ and $ a^j $ ($ 1<i<j $) are accepted by a \textup{0-DHVA} $ H$ and $ \gcd(i,j)=1 $, then $ H $ recognizes $ a^* $.
\end{corollary}

Let us now investigate the case where the set of matrices is commutative.

Let $ L \in \Sigma^* $ be a language. The \textit{commutative closure} of $ L $ is defined as $ com(L)=\{x \in \Sigma^* | \phi(x) \in \phi(L) \} $. A language is \textit{commutative} if $ com(L)=L $.

\begin{theorem}\label{thm: comm}
If a language $L $ is recognized by a \textup{0-NBHVA} $ H $ with a commutative set of matrices, then $L$ is commutative.
\end{theorem}
\begin{proof}
	Let $w \in L  $ and suppose that the string $w=w_{[1]}w_{[2]}\cdots w_{[n]}$ is accepted by $ H $. Let $A_1 A_2\cdots A_n$ be the product of the matrices labeling the computation such that $$v_0 A_1 A_2 \cdots  A_n=v_0$$ where $v_0$ is the initial vector of $ H $. Since the matrices are commutative, then for any permutation $\tau$, $$A_1A_2 \cdots  A_n=A_{\tau(1)} A_{\tau(2)} \cdots  A_{\tau(n)}.$$
	This leads to the acceptance of the string $w'=w_{[\tau(1)]} w_{[\tau(2)]} \cdots w_{[\tau(n)]} $ since $$v_0A_{\tau(1)} A_{\tau(2)} \cdots A_{\tau(n)}=v_0.$$ Hence, if $w$ is accepted by $ H $, then any string obtained from $ w $ by permuting its letters is also accepted by $ H $. Any string $ x $ with $ \phi(x)=\phi(w) $ is in $ L $ and we conclude that $ L $ is commutative. 
	
\end{proof}

When the computation is not blind, then the class of languages recognized is no longer commutative. The language of balanced strings of brackets $ \dyck $ can be recognized by 0-DHVA(1) as follows. Starting with the initial vector $ (1) $, for each left bracket the vector is multiplied by $ (2) $. As long as the vector is not equal to (1), for each right bracket, the vector is multiplied by ($\frac{1}{2}$). If the vector is equal to $ (1) $ and the next symbol is a right bracket, then the vector is set to (0).

\begin{corollary}\label{cor: commreverse}
If a language $L $ is recognized by a \textup{0-NBHVA} $ H $ with a commutative set of matrices, then $L=L^r$.	
\end{corollary}
\begin{proof}
	Suppose that $w \in L$. Then it is clear that $w^r$ will be also accepted by $ H $ by Theorem \ref{thm: comm} and $w^r \in L$. Since for every string $ w $ it is true that  $w \in L $ if and only if $ w^r \in L$, we conclude that $L=L^r$.
	
\end{proof}

We now focus on stateless HVAs whose vectors have dimension 1 and demonstrate some results on stateless FAMWs. Note that stateless FAMs do not process the end-marker  $\dollar$ by definition, since their single state is also an accept state and the computation ends once $\dollar$ is scanned in an accept state. 

We start by comparing the class of languages recognized by stateless versions of deterministic and blind finite automata with multiplication, and 1-dimensional homing vector automata. The capability of multiplication with negative rational numbers brings additional power to the stateless DBHVA(1)s. 

\begin{theorem}
	$ \mathfrak{L} \textup{(0-DFAMW)} \subsetneq \mathfrak{L}\textup{(0-DBHVA(1))}.$
\end{theorem}
\begin{proof}Let $ \mathtt{EVENAB}=\{a^nb^n |~n = 2k \mbox{ for some }k\geq 0 \}$. The following 0-DBHVA(1) recognizes $ \mathtt{EVENAB} $: The register is multiplied by $ (-2) $ and $ (\frac{1}{2}) $ when the machine reads an $ a $ and $ b $ respectively. 
	
	Suppose that there exists some  0-1DFAMW  recognizing $ \mathtt{EVENAB} $. Let $ m_a $ and $ m_b $ be the positive rational numbers that are multiplied by the register upon reading $ a $ and $ b $. Since $ aabb\in \mathtt{EVENAB} $, it is true that $ m_a^2m_b^2=1 $. Since both $ m_a $ and $ m _b$ are positive, it is not possible that $ m_am_b=-1 $. It follows that $ m_am_b=1 $, in which case the string $ ab $ is accepted and we get a contradiction. Hence we conclude that $ \mathtt{EVENAB} $ cannot be recognized by any 0-DFAMW.
\end{proof}

When the register is multiplied with only positive rational numbers, then we have $  \mathfrak{L} $(0-DFAMW)=$\mathfrak{L}$(0-DBHVA(1))$ _\mathbb{Q^+} $. The equivalences also hold for the nondeterministic models. 

For real-time, nondeterministic and blind machines, we can state the following theorem.

\begin{theorem}\label{thm: commu}
	If $ L \in \mathfrak{L}(\textup{0-NFAMW})$, then $ L $ is commutative.
\end{theorem}
\begin{proof}
	A 0-NFAMW is a 0-NBHVA(1) whose register is multiplied with only positive rational numbers. Therefore, $ L $ is also accepted by a 0-NBHVA(1). Since multiplication in dimension 1 is commutative, the result follows by Theorem \ref{thm: comm}.
\end{proof}

It is known that a bounded language is accepted by a 1NFAMW iff it is semilinear \cite{ISK76}. In the next theorem, we prove a similar result and characterize the class of languages recognized by 0-DFAMWs. We show that any language recognized by a 0-DFAMW is commutative and semilinear. Furthermore, any commutative language whose Parikh image is the set of nonnegative solutions to a system of linear homogeneous Diophantine equations can be recognized by a 0-DFAMW. 

\begin{theorem}
 $ L \in \mathfrak{L}(\textup{0-DFAMW}) $ iff $ L $ is commutative and $\phi(L) $ is the set of nonnegative integer solutions to a system of linear homogeneous Diophantine equations.
\end{theorem}
\begin{proof} Let  $ L $ be a language over the alphabet $ \Sigma=\{\sigma_1,\dots,\sigma_n\}  $ recognized by a 0-DFAMW $ V $. Let $M=\{m_1,m_2,\dots,m_n\}$ be the set of rational numbers such that the register is multiplied with $m_i$ upon reading $\sigma_i$. Let $P=\{p_1,p_2,\dots ,p_k\}$ be the set of prime factors of the denominators and the numerators of the rational numbers in $M$. Then each $ m_i $ can be expressed as 
	$$
	m_i=\frac{p_1^{x_{1_i}}p_2^{x_{2_i}}\cdots p_k^{x_{k_i}}}{p_1^{y_{1_i}}p_2^{y_{2_i}}\cdots p_k^{y_{k_i}}} .
	$$  
	
	If a string $w$ is accepted by $V$, then the value of the register should be equal to 1 after reading $ w $, which is possible only if
	$$
	m_1^{w_{|\sigma_1|}}m_2^{w_{|\sigma_2|}}\cdots m_n^{w_{|\sigma_n|}}=1.
	$$
	This leads to the following system of linear Diophantine equations in $ n$ variables.
	
	\begin{align*}
	(x_{1_1}-y_{1_1})w_{|\sigma_1|}+(x_{1_2}-y_{1_2})w_{|\sigma_2|}+\dots + (x_{1_n}-y_{1_n})w_{|\sigma_n|}&=0\\
	(x_{2_1}-y_{2_1})w_{|\sigma_1|}+(x_{2_2}-y_{2_2})w_{|\sigma_2|}+\dots + (x_{2_n}-y_{2_n})w_{|\sigma_n|}&=0\\
	&\vdots\\
	(x_{k_1}-y_{k_1})w_{|\sigma_1|}+(x_{k_2}-y_{k_2})w_{|\sigma_2|}+\dots + (x_{k_n}-y_{k_n})w_{|\sigma_n|}&=0\\
	\end{align*}
	
	
	For $ j=1,\dots ,k $, the $ j $'th equation is stating that the exponent of $ p_j $ is equal to 0 after reading $ w$. 

	Hence, we can conclude that the Parikh images of the accepted strings are the nonnegative solutions to a system of linear homogeneous Diophantine equations. $ L $ is commutative by Theorem \ref{thm: commu}. (One can further conclude that $ L $ is semilinear.)  
	
	For the converse, suppose that we are given a commutative language $ L $ over the alphabet $ \Sigma = \{\sigma_1,\dots,\sigma_n\} $. Let $ S $ be the set of Parikh images of the strings in $ L $. $ S $ is the set of nonnegative solutions to a system of, say, $ k $ linear homogeneous Diophantine equations in $ n $ unknowns,

\begin{align*}
a_{11}s_1 + a_{12}s_2+\dots +a_{1n}s_n&=0\\
a_{21}s_2 + a_{22}s_2+\dots +a_{2n} s_n&=0\\
&\vdots\\
a_{k1}s_1+a_{k2}s_2+\dots + a_{kn}s_n&=0\\
\end{align*} 
where $ \mypar{s_1 ~~~s_2 ~~~\dots~~~ s_n} \in S$. 

We construct a 0-DFAMW $ V $ recognizing $ L $ as follows. We choose a set of $ k $ distinct prime numbers, $ P=\{p_1,p_2,\dots,p_k\} $. When $ V$ reads $ \sigma_i $, the register is multiplied by $$ m_i= p_1^{a_{1i}}p_2^{a_{2i}}\cdots p_k^{a_{ki}}.$$ Suppose that a string $ w $ is accepted by $ V $. Then 
	$$
	m_1^{w_{|\sigma_1|}}m_2^{w_{|\sigma_2|}}\cdots m_n^{w_{|\sigma_n|}}=1.
	$$
	The product is equal to 1 if all of the exponents of the prime factors are equal to 0, that is when $ \mypar{w_{|\sigma_1|} ~~~w_{|\sigma_2|}~~~ \dots ~~~ w_{|\sigma_n|} }\in S$. Hence we see that the set of accepted strings are those with Parikh image in $ S $. Since $ L $ is commutative, any string $ w $ with $ \phi(w) \in S $ belongs to $ L $ and we conclude that $ V $ recognizes $ L $.
	
\end{proof}

Note that $ \mathfrak{L} $(0-DFAMW)=$\mathfrak{L}$(0-1DFAMW), since a 0-1DFAMW that has an instruction to stay on some input symbol cannot read the rest of the string.

\subsection{Regular languages}

Let $ D $ be an $ n $-state deterministic finite automaton. Without loss of generality, we can enumerate its states as $ q_1,\ldots,q_n $ where $q_1$ is the initial state. Moreover we can denote $ q_i $ by $ e_i $, which is the basis vector in dimension $ n $ having 1 in its $ i $'th entry, and 0 otherwise. Besides, for each symbol $\sigma $, we can design a zero-one matrix, say $ A_\sigma $, that represents the transitions between the states, i.e. $ A_\sigma(i,j) = 1 $ if and only if $ D $ switches from $ q_i $ to $ q_j $ when reading symbol $\sigma $. Thus, the computation of $  D $ on an input, say $ w $, can be traced by matrix-vector multiplication:  
\[
e_j =  e_1 A_{w[1]} A_{w[2]} \cdots A_{w[|w|]}
\]
if and only if $  D $ ends its computation in $ q_j $ after reading $ w $.

Based on this representation, we can easily observe that if a language $ L $ is recognized by an $ n $-DFA whose initial state is the single accept state, then $ L $ is recognized by a 0-DBHVA($n$).

Let us give some examples of stateless HVAs recognizing regular languages.
\\

\noindent \textbf{Example 1:} \label{ex: 2} Let $\mathtt{AB_k}^*=\{a^kb^k\}^*$ for a given $k>1$, and let us construct a 0-DBHVA(2$ k $)  $ H_k $ recognizing $\mathtt{AB_k}^*$. The initial vector of $ H_k  $ is $ v_0=(1 ~~ 0 ~~ \cdots ~~ 0) $. 
For each $ a $, the value in the $ i $'th entry of the vector is transferred to the $ (i+1) $'st entry when $ 1 \leq i \leq k $, and, the rest of the entries are set to 0. For each $ b $, the value in the $ (i+k) $'th entry of the vector is transferred to the $ (i+k+1 \mod 2k) $'th entry, and the rest of the entries are set to 0, ($ 1 \leq i \leq k $). Thus, we return  to the initial vector if and only if after some number of $ (a^kb^k) $ blocks have been read.
\\

\noindent \textbf{Example 2:} \label{ex: 1}The language $ \modm $ is defined as
$
\modm = \{ a^i \mid i \mod m \equiv 0 \}.
$ 
It is easy to observe that any unary $n$-state DFA whose initial state is the single accept state can recognize either $ L_\varepsilon $ or $ \modm $ for some $m\leq n$. Hence, for any $ m>0 $, the language $ \modm  $ is recognized by a \textup{0-DBHVA($m$)}. 
\\

Note that if it is allowed to use arbitrary algebraic numbers in the transition matrices, then for every $ m>0 $, the language $ \modm $ is recognized by a \textup{0-DBHVA(2)} with initial vector $ v_0 = \mypar{1 ~~ 0} $ that multiplies its vector with the matrix $$ A_m = \mymatrix{cc}{ \cos  \frac{2 \pi}{m} & -\sin  \frac{2 \pi}{m} \\ \sin  \frac{2 \pi}{m} & \cos  \frac{2 \pi}{m} } $$ for each scanned symbol. (The entries of $A_m$ are algebraic for any $m$ \cite{Le93}.)    
\\

One may ask whether all regular languages $ L $ satisfying $ L= L^* $ can be recognized by stateless HVAs. We provide a negative answer for the deterministic model. The language $ \modtwothree $ is defined as $ \{ a^2,a^3 \}^* = a^* \setminus \{a\} $, satisfying that $ \modtwothree = \modtwothree^* $. $ \modtwothree $ cannot be recognized by any \textup{0-DHVA} since a 0-DHVA which accepts $ a^2 $ and $ a^3 $, also accepts $ a $ by Lemma \ref{lemma: diff}.

\subsection{Additional results on stateless homing vector automata}

One sees that the nonregular language $\mathtt{AB} = \{ a^nb^n | n \geq 0 \} $ cannot be recognized by any 0-NHVA($ k $) for any $ k $, since $ \mathtt{AB} \neq \mathtt{AB}^* $. On the other hand, $ \Leq = \{ x \in \{a,b\}^* \mid |x|_a = |x|_b \} $ can be recognized by a 0-DBHVA(1) with initial vector $  (1) $, and transition matrices $ A_a = (2) $ and $ A_b = \myvector{\frac{1}{2}} $. In the next theorem, we show that it is possible to recognize the star of $\mathtt{AB} $, even in the deterministic and blind computation mode with a stateless homing vector automaton. Note that $ \mathtt{AB}^* $ cannot be recognized by any 1NFAMW \cite{ISK76}.

\begin{theorem}
	\label{thm:0HVA-upal}
	The language $\mathtt{AB}^*= \{ \{a^nb^n\}^* \mid n \geq 0 \}$ is recognized by a \textup{0-DBHVA(10)}.
\end{theorem}
\begin{proof} We are going to construct a \textup{0-DBHVA(10)} recognizing  $\mathtt{AB}^*$. We start with a 0-DBHVA(3) named $H $, whose initial vector is $ v_0 = (1 ~~ 0 ~~ 0) $. When reading a symbol $a$, $ H $ applies the matrix 
	$$ A_a =  \mymatrix{ccc}{1 & 1 & 0  \\ 0 & 2 & 0  \\ 0 & 0 & 0},
	$$  
	which maps $ ( 1 ~~ t_2 ~~ t_3 ) $ to $ ( 1 ~~ 2t_2 + 1 ~~ 0 ) $. This operator encodes the count of $ a$'s in the value of the second entry. Started with the initial vector, if $ H $ reads $ i $ $a$'s, the vector becomes $ ( 1 ~~ 2^i -1 ~~ 0 ) $.
	
	When reading a symbol $ b $, $ H $ applies the matrix 
	$$ A_b = \mymatrix{rrr}{1 & 0 & -\frac{1}{2} \\ 0 & 0 & \frac{1}{2}  \\  0 & 0 & \frac{1}{2}},
	$$
	that  sets the value of the second entry to 0 and decrements the counter mentioned above, which it maintains in 
	the third entry. A vector starting as $ ( 1 ~~ 2^i -1 ~~ 0 ) $
	would get transformed to  $ ( 1 ~~ 0 ~~ 2^{i-j}-1 ) $
	after reading $ j $  $ b $'s.
	
	Note that the third entry is different than 0 when the number of $ a $'s and $ b $'s do not match in a block. If symbol $ a $ is read after a block of $ b $'s, the counting restarts from 0, and the value of the third entry is again set to 0.
	
	%
	By tensoring $ H $ with itself, we obtain a 0-DBHVA(9). The  initial vector of this new machine is $ v_0 \otimes v_0 $, and its operators are $ A_a \otimes A_a $ and $ A_b \otimes A_b $.
	
	Based on this tensor machine, we define our final 0-DBHVA(10), say, $ H' $. The initial vector is $v_0'= (v_0 \otimes v_0 ~~ 0) = (1 ~~ 0 ~~ \cdots ~~ 0).$
	
	
	$ H' $ applies the matrices $A_a'$ and $ A_b' $ for each scanned $ a $ and $ b $, respectively:
	\[
	A_a' = \mymatrix{c|c}{ A_a \otimes A_a & \begin{array}{c}0 \\ \vdots \\ 0 \\ 1 \end{array} \\ \hline 0 ~~ \cdots ~~ 0 & 1  }~~~A_b' = \mymatrix{c|c}{ A_b \otimes A_b & \begin{array}{c}0 \\ \vdots \\ 0\\0 \end{array} \\ \hline 0 ~~ \cdots ~~ 0 & 1  }.
	\]
	The matrix $ A_a' $ implements $ A_a \otimes A_a $ in its first nine entries. Additionally, it adds the value of the ninth entry of the vector to the tenth one. The operator $ A_b' $ implements $ A_b \otimes A_b $ in its first nine entries. Additionally, it preserves the value of the tenth entry. Suppose that the vector of $ H $ is equal to $ (1~~ a_2 ~~ a_3) $ at some point of the computation. Due to the property of the tensor product, the vector entries of $ H' $ hold the values $ (1 ~~a_2~~ a_3~~ a_2~~ a_2a_2~~ a_2a_3~~ a_3~~ a_3a_2~~ a_3a_3~~ 1) $.

	The value of the third entry is important as it becomes nonzero when the number of $ a $'s and $ b $'s are different in a block. Hence, the ninth entry holds some positive value when this is the case. When $ H' $ starts reading a new block of $ a $'s, it adds the ninth entry to the tenth one with the help of the matrix $ A_a' $. Once this entry is set to a positive value, it never becomes 0 again, since its value can only be changed by the addition of the ninth entry, which is always a square of a rational number. Therefore, such an input is never accepted.
	
	If this is not the case and the number of $ a $'s and $ b $'s are equal to each other in every block, then the third and therefore the ninth entry are always equal to 0 before $ H' $ starts reading a new block. We can conclude that the value of the tenth entry will be equal to 0 at the end of the computation. As the input string ends with a $ b $, the second entry is also set to 0 upon multiplication with $ A_b' $. Hence, all the entries except the first one are equal to 0 and the vector is equal to its initial value which leads to the acceptance of the input string. 
	
	Thus, the only accepted inputs are the members of $ \mathtt{AB}^* $.
\end{proof}

Nondeterministic HVAs are more powerful than their deterministic variants in terms of language recognition in general. In the next theorem, we show that this is also true for the stateless models.

\begin{theorem}
	\begin{enumerate}[i.]
		\item $\mathfrak{L} 
		\textup{(0-DBHVA)} \subsetneq  \mathfrak{L} \textup{(0-NBHVA)}. $
		\item $  \mathfrak{L} 
		\textup{(0-DHVA)} \subsetneq  \mathfrak{L} \textup{(0-NHVA)}. $
	\end{enumerate}
\end{theorem}
\begin{proof}
	Let us construct a 0-NBHVA(1) $ H $ recognizing $\mathtt{LEQ}=\{x \in \{a,b\}^* \mid |x|_a \leq |x|_b\}$. Starting with the initial vector $(1) $, $ H $ multiplies its vector with $ A=(2) $ for each $a$ and with $B_1=(\frac{1}{2})$ or $B_2=(1) $ for each $b$ nondeterministically.
	
	Suppose that $\mathtt{LEQ}$ can be recognized by a 0-DHVA($k$) $ H' $. The strings $ w_1=b $ and $ w_2=ba $ are accepted by $ H' $. By Lemma \ref{lemma: diff}, the string $w_3=a$ is also accepted by $ H' $. We obtain a contradiction, and conclude that $\mathtt{LEQ}$ cannot be recognized by any 0-DHVA($k$).
\end{proof}

Let us look at some closure properties for the stateless models. All of the stateless models are closed under the star operation since for any language recognized by a stateless homing vector automaton, it is true that $ L=L^* $.

\begin{theorem}
	i. $ \mathfrak{L}\textup{(0-DBHVA)} $ is closed under the following operation:\\
	\indent	 a) intersection\\
	ii. $ \mathfrak{L}\textup{(0-DBHVA)} $ and $ \mathfrak{L}\textup{(0-DHVA)} $ are not closed under the following operations:\\	
	\indent			a) union\\
	\indent			b) complement\\
	\indent			c) concatenation
\end{theorem}
\begin{proof} 
	\textit{i. a)} The proof in \cite{SDS16} that DBHVAs are closed under intersection is also valid for stateless models.\\
	\textit{ii. a)} The languages $ \mathtt{MOD_2}$ and $ \mathtt{MOD_3}$ are recognized by 0-DBHVAs by Example \ref{ex: 1}. Their union cannot be recognized by any 0-DHVA, since a 0-DHVA accepting the strings $a^2$ and $a^3$ should also accept the non-member string $a$ by Lemma \ref{lemma: diff}.\\
	\textit{b)} Complement of the language $ \modm $ ($m>1$)is not recognized by any 0-NHVA, since $ \overline{\modm} $ contains $ a $ and any 0-NHVA accepting $ a $ accepts any member of $a^*$.\\
	c) The languages $ \mathtt{MOD_2}$ and $ \mathtt{MOD_3}$ are recognized by 0-DBHVAs by Example \ref{ex: 1}. Their concatenation $\modtwothree$ cannot be recognized by any 0-DHVA.~~~
\end{proof}

$ \mathfrak{L}\textup{(0-NBHVA)} $ is closed under intersection.  $ \mathfrak{L}\textup{(0-NBHVA)} $ and $ \mathfrak{L}\textup{(0-NHVA)} $ are not closed under complement. The proofs are identical.

\section{Open questions and future work}

We proved that 1NBHVAs are more powerful than extended finite automata when both are defined over $ 2 \times 2 $ integer matrices. Is this result still true when both models are defined over $ 3 \times 3  $ integer matrices?

Do 0-NHVAs recognize every regular language $L$ satisfying $ L=L^*$? Is there any nonregular language $L$ satisfying $L=L^*$ that cannot be recognized by any stateless HVA?
 
We proved that any language recognized by a 0-NFAMW is commutative. What can we say about the non-blind case?

We gave a characterization for the class of languages recognized by 0-DFAMW. Can we give a similar characterization for the non-blind and nondeterministic models?

\section*{Acknowledgments} This research was supported by Bo\u{g}azi\c{c}i University Research Fund (BAP) under grant number 11760. Salehi is partially supported by T\"{U}B\.{I}TAK (Scientific and Technological Research Council of Turkey). Yakary{\i}lmaz is partially supported by ERC Advanced Grant MQC. We thank Flavio D'Alessandro for his helpful answers to our questions, and the anonymous reviewers for their constructive comments.

\bibliographystyle{splncs03}
\bibliography{references,tcs}

\end{document}